\documentclass[letterpaper, 10 pt, conference]{ieeeconf}
\IEEEoverridecommandlockouts 
\let\labelindent\relax
\usepackage{amsmath, amsthm, amssymb}
\usepackage{color}
\usepackage{enumitem}
\usepackage{graphicx}
\usepackage[english]{babel}
\usepackage[autostyle]{csquotes}
\usepackage{setspace}
\usepackage{titlesec}

\theoremstyle{definition}
\newtheorem{theorem}{Theorem}
\newtheorem{proposition}{Proposition}

\newtheorem{corollary}{Corollary}
\newtheorem{definition}{Definition}
\theoremstyle{remark}
\newtheorem*{remark}{Remark}

\title{\LARGE \bf
Stackelberg Routing of Autonomous Cars in Mixed-Autonomy Traffic Networks
}

\author{Maxwell Kolarich$^{1}$ and Negar Mehr$^{2}$
\thanks{$^{1}${Maxwell Kolarich is with Aerospace Engineering Department of the University of Illinois, Urbana Champaign, 139 Coordinated Science Laboratory, IL, USA
	{\tt\small mak13@illinois.edu}}
}
\thanks{$^{2}$Negar Mehr is with the Aerospace Engineering Department of the University of Illinois, Urbana Champaign, 161 Coordinated Science Laboratory, IL, USA
        {\tt\small negar@illinois.edu}}%
}

\newcommand{\network}{\mathcal{G}}

\newcommand{\nodeset}{\mathcal{N}}
\newcommand{\link}{l}
\newcommand{\linkset}{\mathcal{L}}
\newcommand{\OD}{w}
\newcommand{\ODset}{\mathcal{W}}
\newcommand{\path}{p}
\newcommand{\pathset}{\mathcal{P}}
\newcommand{\demand}{r}
\newcommand{\latency}{e}
\newcommand{\cost}{C}

\newcommand{\auto}{f^a}
\newcommand{\human}{f^h}
\newcommand{\stack}{s}
\newcommand{\induced}{t}
\newcommand{\autoopt}{f^{a^*}}
\newcommand{\humanopt}{f^{h^*}}
\newcommand{\degauto}{\alpha}
\newcommand{\degasym}{\mu}

\usepackage[colorlinks]{hyperref}

\begin{document}

\abovedisplayskip 1ex plus3pt minus1pt%
\belowdisplayskip \abovedisplayskip%
\abovedisplayshortskip 0pt plus3pt%
\belowdisplayshortskip 1ex plus3pt minus1pt
\abovecaptionskip=0pt
\belowcaptionskip=0pt

\maketitle
\thispagestyle{plain}
\pagestyle{plain}

\begin{abstract}
As autonomous cars are becoming tangible technologies, road networks will soon be shared by human-driven and autonomous cars. However, humans normally act selfishly which may result in network inefficiencies. In this work, we study increasing the efficiency of mixed-autonomy traffic networks by routing autonomous cars altruistically. We consider a Stackelberg routing setting where a central planner can route autonomous cars in the favor of society such that when human-driven cars react and select their routes selfishly, the overall system efficiency is increased. We develop a Stackelberg routing strategy for autonomous cars in a mixed-autonomy traffic network with arbitrary geometry. We bound the price of anarchy that our Stackelberg strategy induces and prove that our proposed Stackelberg routing will reduce the price of anarchy, i.e. it increases the network efficiency. Specifically, we consider a non-atomic routing game in a mixed-autonomy setting with affine latency functions and develop an extension of the SCALE Stackelberg strategy for mixed-autonomy networks. We derive an upper bound on the price of anarchy that this Stackelberg routing induces and demonstrate that in the limit, our bound recovers the price of anarchy bounds for networks of only human-driven cars. 
\end{abstract}

\section{Introduction}
As autonomous cars are becoming tangible technologies, traffic networks will soon experience a transient era, when both human-driven and autonomous cars shared roads, i.e. traffic networks will have mixed vehicle autonomy. Consequently, it is crucial to study the mobility impacts of autonomous cars under mixed vehicle autonomy. In~\cite{wu2017emergent}, deep reinforcement learning was utilized to predict the emergent behavior of mixed-autonomy networks. It was shown in~\cite{li2019extended,mehr2021game} that autonomous cars can select their lanes such that the overall mobility is improved. Additionally, it has been demonstrated that autonomous cars can facilitate vehicle {platooning} using their connectivity, that is, autonomous cars can form groups of vehicles that maintain a headway that is shorter than the average headway of human-driven cars. As a result, it is expected that autonomous cars can increase road capacities in mixed-autonomy networks. In~\cite{lioris2017platoons}, it was shown that the capacity increases of vehicle platooning can be up to three-fold. Such capacity increases were modeled for mixed-autonomy networks in~\cite{lazar2017capacity}.

A key challenge in predicting the mobility benefits of autonomous cars in mixed-autonomy networks is to account for how human-driven cars will respond and adapt to the presence of autonomy. Humans normally act selfishly and decide in their own best interest. This results in a complex interaction between humans and autonomous cars which may even have undesirable impacts on society. For instance, in our previous work~\cite{mehr2018can,mehr2019will}, we showed that the capacity increases of autonomous cars may even worsen traffic congestion when both human-driven and autonomous cars select their routes selfishly. The inefficiency that result from selfish route choices of travellers is commonly measured via the notion of price of anarchy~\cite{roughgarden2002selfish} which measures how far from the optimal the performance of a network is due to selfish routing. In~\cite{lazar2018price} it was shown that the price of anarchy of networks with mixed autonomy is indeed larger than that of networks with only human-driven cars. This indicates that the inefficiencies that result from selfish route choices can be higher for mixed-autonomy networks. Recent works have considered decreasing the inefficiency of mixed-autonomy networks through pricing and tolling~\cite{mehr2019pricing,lazar2019optimal}.

We consider increasing the efficiency of mixed-autonomy traffic networks by routing autonomous cars \textit{altruistically}. We consider the setting where a central planner can route autonomous cars in the favor of society such that when human-driven cars react and select their routes selfishly, the overall system efficiency is increased. Such routing strategies have been extensively studied in the routing games literature as Stackelberg strategies~\cite{roughgarden2004stackelberg,krichene2014stackelberg,BONIFACI2010,harks2010stackelberg,karakostas2009stackelberg,korilis1997achieving}, where several Stackelberg strategies have been proposed for various network geometries and latency functions. It has been proved that the price of anarchy of the system decreases as a result of Stackelberg routing of a fraction of the flow. However, all these works considered a routing game with one single class of vehicles in the network. These results cannot be directly applied to mixed-autonomy networks where there exist two classes of vehicles (human-driven and autonomous). Moreover,  in a mixed-autonomy setting, the two classes of vehicle contribute asymmetrically to the latency along link; the autonomous cars create less externalities due to their platooning capabilities. Only recently, in~\cite{biyik2018altruistic}, Stackelberg routing of autonomous cars in mixed-autonomy networks was studied for networks with parallel links. 

In this paper, we develop Stackelberg routing strategies for autonomous cars in a mixed-autonomy traffic network with arbitrary geometry. We bound the price of anarchy that our Stackelberg strategy induces and prove that our proposed Stackelberg routing will reduce the price of anarchy, i.e. it increases the network efficiency. Specifically, we consider a non-atomic routing game in a mixed-autonomy setting with affine latency functions. We consider a setting where an $\alpha$ fraction of demand along each origin/destination (O/D) pair is autonomous and can be routed by a central planner altruistically. We develop an extension of the SCALE Stackelberg strategy~\cite{BONIFACI2010} for mixed-autonomy networks and derive an upper bound on the price of anarchy that this Stackelberg routing induces. We demonstrate that in the limit, when the fraction of autonomous cars is negligible, our bound recovers the price of anarchy bounds for networks of only human-driven cars. 
 
The organization of this paper is as follows. In Section~\ref{sec:mixed-autonomy-routing-game}, we describe our mixed-autonomy routing game. In Section~\ref{sec:prior-work}, we review some known results from prior work that we utilize in this paper. We extensively discuss and prove the effectiveness of our proposed Stackelberg routing strategy in Section~\ref{sec:stackelberg}. Then, we conclude the paper and provide future directions in Section~\ref{sec:conclusion}.

\section{Mixed-Autonomy Routing Games}\label{sec:mixed-autonomy-routing-game}
We consider a traffic network as a directed graph $\network = (\nodeset, \linkset)$ with a set nodes $\nodeset$ and links $\linkset$. Each link $\link \in \linkset$ is a pair of unique elements $u,v \in \nodeset$ representing a directed edge from node $u$ to node $v$. We consider a network with \textit{arbitrary} geometry. Additionally, we assume we have a set $\ODset$ of $k$ origin/destination (O/D) pairs where each O/D pair is denoted by $(o_\OD, d_\OD)$, with $o_\OD, d_\OD \in \nodeset$ and $o_\OD \neq d_\OD$ representing the origin and destination nodes respectively. For every $\OD \in \ODset$, we assume that a fixed demand $\demand_\OD > 0$ is given which specifies the total amount of required flow (including both human-driven and autonomous cars) that needs to be routed along O/D pair $w$. 
We let $\demand = (\demand_\OD)_{\OD \in \ODset}$ be the vector of network demands. Let $\pathset_\OD$ be the set of all paths that connect O/D pair $\OD$. For each O/D pair $\OD \in \ODset $, we assume there exists at least one path from its origin to its destination i.e., $\pathset_\OD \neq \emptyset$. Additionally, let $\pathset = \cup_{\OD \in \ODset} \pathset_\OD$ denote the set of all network paths.

For each route $p \in \pathset$, we denote the flow of autonomous and human-driven cars along $p$ by $\auto_\path$ and $\human_\path$, respectively. Moreover, we define the network flow vectors $\auto = (\auto_\path)_{\path \in \pathset}$ and $\human = (\human_\path)_{\path \in \pathset}$.
For each O/D pair $\OD \in \ODset$, we let $\alpha_\OD$ be the given autonomy fraction along O/D pair $\OD$ which is the fraction of the demand along $\OD$ that is autonomous. 
For a given network, the flows $\auto$ and $\human$ are feasible if they satisfy flow conservation i.e, for all $\OD \in \ODset$, we have
\begin{equation} \label{eq:feasibility}
    \sum_{\path \in \pathset_\OD} \auto_\path = \degauto_\OD \demand_\OD, \quad \quad \sum_{\path \in \pathset_\OD} \human_\path = (1-\degauto_\OD) \demand_\OD.
\end{equation}

For each link $\link \in \linkset$, the flows of autonomous and human-driven cars along the link are defined as:
\begin{equation}
    \auto_\link := \sum_{\path \in \pathset : \link \in \path} \auto_\path \quad,  \quad \human_\link := \sum_{\path \in \pathset : \link \in \path} \human_\path.
\end{equation}
Every link $\link \in \linkset$ has an associated \textit{latency function} $\latency_\link: \mathbb{R}^2_+ \to \mathbb{R}_+$ which is a nondecreasing and differentiable function of flow along the link. We define $\latency = (\latency_\link)_{\link \in \linkset}$ as the vector of latency functions. 
Additionally, for each path $\path \in \pathset$, we define the latency along $p$ as 
\begin{equation}
    \latency_\path(\auto, \human) := \sum_{\link \in \linkset : \link \in \path} \latency_\link(\auto_\link, \human_\link).
\end{equation}

We assume that autonomous vehicles are capable of \textit{platooning}, in which groups of vehicles may travel together in order to maintain shorter headways. Platooning may be used to directly increase the capacity of autonomous vehicles on a given link as shown in~\cite{lazar2017capacity} and~\cite{SALA2020163}. Under the platooning assumption, we take our latency functions to be affine functions of the form presented in~\cite{mehr2019will}:
\begin{equation} \label{eq:latency_functions}
    \latency_\link(\auto_\link, \human_\link) = a_\link \auto_\link + h_\link \human_\link + b_\link
\end{equation}
where $a_\link,h_\link > 0$, and $b_\link \geq 0$. Thus, the relative impact of autonomous and human-driven vehicles is determined by the coefficients $a_\link$ and $h_\link$. This motivates the definition of the \textit{degree of asymmetry}:
\begin{equation} \label{eq:degree_asymmetry}
    \degasym_\link := \frac{a_\link}{h_\link}
\end{equation}
Since autonomous vehicles are capable of maintaining shorter headways, we assume $\degasym_\link \in (0,1]$.

The tuple $(\network, \demand, \latency)$ is called an \textit{instance} of a mixed autonomy traffic routing game. The \textit{social cost} of flows $\auto,\human$ is given by
\begin{equation}
    \cost(\auto, \human) = \sum_{\link \in \linkset} (\auto_\link + \human_\link) \latency_\link(\auto_\link, \human_\link).
\end{equation}
For a mixed-autonomy routing game $(\network, \demand, \latency)$, the optimal flow vectors ${\auto}^*, {\human}^*$ are the feasible flow vectors which attain the minimum social cost.

In a Stackelberg routing game, we assume: 1) autonomous cars are controlled by a \textit{central planner} attempting to route their portion of flow to minimize the overall social cost and, 2) human-driven cars are represented by \textit{selfish nonatomic agents} that seek to control their infinitesimal portion of the flow demand to minimize their own latency by choosing paths of minimum latency. For this game, in addition to the instance $(\network, \demand, \latency)$, we are given a \textit{network autonomy fraction} $\degauto \in (0,1)$. This parameter represents the fraction of total network demand satisfied by autonomous vehicles and is defined as
\begin{equation}
    \degauto := \frac{\sum_{\OD \in \ODset} \degauto_\OD \demand_\OD}{\sum_{\OD \in \ODset} \demand_\OD}.
\end{equation}
A \textit{Stackelberg strategy} is an autonomous flow vector $\stack$ feasible with respect to the network autonomy fraction:  
\begin{equation}
    \sum_{\OD \in \ODset} \sum_{\path \in \pathset_\OD} \stack_\path = \degauto \sum_{\OD \in \ODset} \demand_\OD.
\end{equation}
If $\degauto_\OD = \degauto$ for all $\OD \in \ODset$, $\stack$ is called a \textit{weak Stackelberg strategy}~\cite{roughgarden2005selfish}. Additionally, a Stackelberg strategy $\stack$ is called \textit{opt-restricted} if for each link $\link \in \linkset$,  $\stack_\link \leq \autoopt_\link + \humanopt_\link$~\cite{correa2007stackelberg}.



Since we assumed that the human-driven cars act selfishly, we assume that human-driven cars reach \textit{Nash Equilibrium}, a distribution of traffic in which no nonatomic player has any incentive to unilaterally change their path. 
\begin{definition} \label{def:nash_induced}
    Given a mixed-autonomy routing game $(\network, \demand, \latency)$, network autonomy fraction $\degauto$, and Stackelberg strategy $\stack$, a feasible flow of human-driven cars $\induced$ is an induced Nash flow if and only if for every O/D pair $\OD \in \ODset$ and every pair of paths $\path, \path' \in \pathset_\OD$
    \begin{equation} 
        \induced_\path \latency_\path(\stack_\path, \induced_\path) \leq \induced_\path \latency_{\path'}(\stack_\path, \induced_\path).
    \end{equation}
\end{definition}

To measure the efficiency of a given Stackelberg routing, we will consider the ratio of the social cost induced by the Stackelberg strategy to the optimal cost which is called the \textit{price of anarchy}, which we denote as
\begin{equation}
    \mathrm{PoA} = \frac{\cost(\stack, \induced)}{\cost(\autoopt, \humanopt)}.
\end{equation}
The price of anarchy captures the inefficiencies that are caused by the selfish routing decisions of human-driven cars. If all cars are routed optimally, then $\mathrm{PoA}=1$. For a non-optimal routing, $\mathrm{PoA}$ measures how far from optimum the network routing is.

\section{Prior Work}\label{sec:prior-work}
We state a number of known results and methods from previous literature that will be utilized to develop and compare with our price of anarchy upper bound. 

For a mixed-autonomy routing game where both autonomous and human-driven vehicles are selfishly routed, it was shown in \cite{lazar2018price} that the price of anarchy can be bound independently of $\degauto$ and in terms of the \textit{minimum degree of asymmetry}:
\begin{equation}
    \degasym := \min_{\link \in \linkset} \degasym_\link.
\end{equation}
The price of anarchy bound was given for $\degasym > 1/4$ as
\begin{equation}
    \mathrm{PoA} \leq \frac{4\degasym}{4\degasym-1}.
\end{equation}

A number of Stackelberg strategies have been previously studied on a wide range of network geometries and latency functions. A popular and intuitive strategy is to route an $\degauto$ fraction of the optimal flow along each link. This is known as the SCALE strategy. It was shown in \cite{BONIFACI2010} that for arbitrary networks with a single class of vehicles in the network $(\degasym = 1)$ that the price of anarchy for the SCALE strategy is bounded by
\begin{equation}
    \mathrm{PoA} \leq \frac{(1+\sqrt{1-\degauto})^2}{2(1+\sqrt{1-\degauto})-1}
\end{equation}

In order to develop an efficient routing strategy for the mixed-autonomy case, we define a mixed-autonomy equivalent to the SCALE strategy as follows.
\begin{definition} \label{def:scale_strategy}
    Given a mixed-autonomy routing game $(\network, \demand, \latency)$ and network autonomy fraction $\degauto$, the mixed-autonomy SCALE strategy is given by
    \begin{equation} \label{eq:scale_strategy}
        \stack := \degauto(\autoopt + \humanopt)
    \end{equation}
    where $\autoopt$ and $\humanopt$ are the optimal flow vectors for $(\network, \demand, \latency)$.
\end{definition}
The mixed-autonomy SCALE strategy is a weak Stackelberg strategy since it assigns autonomous flow to satisfy an $\degauto$ fraction of the demand on each O/D pair. Additionally, since $\degauto \in (0,1)$, the SCALE strategy is clearly an opt-restricted Stackelberg strategy.


A particularly useful method for deriving an upper bound on the price of anarchy for an opt-restricted Stackelberg strategy is the \enquote{$\lambda$-approach} given in~\cite{BONIFACI2010} and~\cite{harks2010stackelberg}. We state this result here, extended to mixed-autonomy routing games, in the following definitions and proposition.

First, for the nonnegative number $\lambda$, define the link based parameter
\begin{equation}
    \omega_\link(\lambda) := \sup_{\autoopt_\link, \humanopt_\link,  \, \induced_\link \geq 0} \frac{\autoopt_\link + \humanopt_\link}{\stack_\link + \induced_\link}  \frac{\latency_\link(\stack_\link, \induced_\link) - \lambda \latency_\link(\autoopt_\link, \humanopt_\link)}{\latency_\link(\stack_\link, \induced_\link)},
\end{equation}
and the maximum value of this parameter over all links in the network as
\begin{equation} \label{eq:omega_full}
    \omega(\lambda) := \max_{\link \in \linkset} \omega_\link(\lambda).
\end{equation}
Additionally, the feasible set of $\lambda$ is defined to be
\begin{equation} \label{eq:Lambda_set}
    \Lambda := \{\lambda \in [0,1] \: | \: \omega(\lambda) < 1 \}.
\end{equation}

\noindent Then, the $\lambda$-approach is stated as follows.
\begin{proposition} [{{\cite{BONIFACI2010,harks2010stackelberg}}}] \label{prop:lambda_method}
    Let $\lambda \in \Lambda$. Then
    \begin{equation}
        \mathrm{PoA} \leq \frac{\lambda}{1-\omega(\lambda)}.
    \end{equation}
\end{proposition}

To simplify the above approach for the mixed-autonomy setting, we introduce some additional definitions. 
First, for each link $\link \in \linkset$, we characterize the ratio of optimal flow to induced flow as
\begin{equation} \label{eq:gamma}
    \gamma_\link := \frac{\autoopt_\link + \humanopt_\link}{\stack_\link + \induced_\link}.
\end{equation}
Similarly we define the ratio of latency functions
\begin{equation} \label{eq:beta}
    \beta_\link := \frac{\latency_\link(\stack_\link, \induced_\link) - \latency_\link(\autoopt, \humanopt)}{\latency_\link(\stack_\link, \induced_\link)}.
\end{equation}
These two definitions allow us to express (\ref{eq:omega_full}) as
\begin{equation} \label{eq:omega_link}
    \omega_\link(\lambda) = \sup_{\gamma_\link \geq 0} \gamma_\link \left[1+(\beta_\link-1)\lambda\right].
\end{equation}

Our approach, which applies to routing games with general mixed-autonomy latency functions, is to derive an upper bound on the price of anarchy for a given Stackelberg strategy as follows:
\begin{enumerate}
    \item Develop an upper bound on the ratio of latency functions $\beta_\link$ as a function of the ratio of flows $\gamma_\link$ dependent only on known link parameters.
    \item Determine the value of $\gamma_\link$ that corresponds to $\sup_{\gamma_\link \geq 0} \gamma_\link \left[1+(\beta_\link-1)\lambda\right]$ to express $\omega_\link$ as a function of $\lambda$.
    \item Calculate values of network parameters that provide the maximum value of $\omega_\link(\lambda)$ over all links, $\omega(\lambda)$.  
    \item Find conditions on network parameters that determine the feasible set $\Lambda$.
    \item Find the minimum value of $\lambda/(1-\omega(\lambda))$ achieved by $\lambda \in \Lambda$.  
\end{enumerate}

\section{Mixed-Autonomy Stackelberg Routing}\label{sec:stackelberg}
For this study, we restrict Stackelberg routing strategies to the class of affine mixed autonomy latency functions as given in~\eqref{eq:latency_functions}. We consider our Stackelberg strategy to be the mixed-autonomy SCALE strategy given by~\eqref{eq:scale_strategy} and derive upper bounds for on the price of anarchy for this strategy in a mixed-autonomy routing game. To this end, we utilize the $\lambda$-approach stated in Proposition \ref{prop:lambda_method} and its extension to the mixed-autonomy setting. We present our Propositions and Theorems in an order consistent with the steps described at the end of Section~\ref{sec:prior-work}.

\begin{proposition} \label{prop:beta_link}
For a mixed-autonomy routing game $(\network, \demand, \latency)$ with affine latency functions, network autonomy fraction $\degauto$, and mixed-autonomy SCALE strategy $\stack = \degauto(\autoopt + \humanopt)$, we have
\begin{equation} \label{eq:beta_link}
    \beta_\link(\gamma_\link) \leq 
    \begin{cases}
        \hfil 1 &
        \mathrm{for} \quad \gamma_\link = 0 \\
        \hfil 1 - \frac{1 - \degauto^*_\link(1-\degasym_\link)}{\gamma_\link^{-1} - \degauto(1-\degasym_\link)} &
        \mathrm{for} \quad 0 < \gamma_\link < \tilde{\gamma}_\link \\
        \hfil 0 &
        \mathrm{for} \quad \tilde{\gamma}_\link \leq \gamma_\link \leq 1/\degauto
    \end{cases}
\end{equation}
where $\tilde{\gamma}_\link = [1 + (\degauto - \degauto^*_\link)(1-\degasym_\link)]^{-1}$ and $\degauto^*_\link = \autoopt_\link/(\autoopt_\link+\humanopt_\link)$.
\end{proposition}
\begin{proof}

We prove this using the special structure of our latency functions. Using~\eqref{eq:latency_functions} with the definitions~\eqref{eq:beta} and~\eqref{eq:degree_asymmetry}, we have
\begin{equation}
    \beta_\link = \frac{h_\link[(\degasym_\link\stack_\link+\induced_\link) - (\degasym_\link\autoopt_\link+\humanopt_\link)]}{h_\link(\degasym_\link\stack_\link+\induced_\link) + b_\link}.
\end{equation}

First, let us consider $\gamma_\link = 0$. In this case, $\autoopt_\link + \humanopt_\link = 0$; therefore, $\autoopt_\link= \humanopt_\link = 0$ since the vehicle flows are nonnegative. Thus, if we choose $\stack_\link = \degauto(\autoopt_\link+\humanopt_\link)$ we have $\stack_\link = 0$ as well and $\beta(0) = h_\link\induced_\link/(h_\link\induced_\link+b_\link) \leq 1$.

When $\gamma_\link \neq 0$ the ratio of latency functions $\beta_\link$ may be piecewise bound by considering the sign of the numerator argument, that is, by determining when $(\degasym_\link\stack_\link+\induced_\link) - (\degasym_\link\autoopt_\link+\humanopt_\link) = 0$. We can rewrite $\degasym_\link\stack_\link + \induced_\link = \degasym_\link\autoopt_\link + \humanopt_\link$ as
\begin{equation}
    (\stack_\link+\induced_\link) - (1-\degasym_\link)\stack_\link = (\autoopt_\link+\humanopt_\link) - (1-\degasym_\link)\autoopt_\link.
\end{equation}
Additionally, $\autoopt_\link + \humanopt_\link \neq 0$ since $\gamma_\link \neq 0$. Using the definition of $\gamma_\link$~\eqref{eq:gamma}, the above statement implies
\begin{equation}
    (\autoopt_\link+\humanopt_\link)/\gamma_\link - (1-\degasym_\link)\stack_\link = (\autoopt_\link+\humanopt_\link) - (1-\degasym_\link)\autoopt_\link
\end{equation}
and the numerator of $\beta_\link$ is equal to zero for
\begin{equation}
    \tilde{\gamma} := \left[1 + \frac{\stack_\link-\autoopt_\link}{\autoopt_\link+\humanopt_\link}(1-\degasym_\link)\right]^{-1}.
\end{equation}
Thus, if we take our Stackelberg strategy to be $\stack_\link = \degauto(\autoopt_\link + \humanopt_\link)$ we have $\tilde{\gamma}_\link = [1 + (\degauto - \degauto^*_\link)(1-\degasym_\link)]^{-1}$ where  $\degauto^*_\link = \autoopt_\link/(\autoopt_\link+\humanopt_\link)$ is the \textit{optimal flow degree of autonomy} on link $\link$. We now proceed to develop bounds on $\beta_\link$ based on the specific regions of $\gamma_\link$ relative to $\tilde{\gamma}_\link$:

\textbf{Case 1.} If $0 < \gamma_\link < \tilde{\gamma}_\link$, or equivalently $(\degasym_\link\stack_\link+\induced_\link) - (\degasym_\link\autoopt_\link+\humanopt_\link) > 0$, since $h_\link>0, b_\link \geq 0$, it follows that
\begin{equation}
    \beta_\link \leq \frac{(\degasym_\link\stack_\link+\induced_\link) - (\degasym_\link\autoopt_\link+\humanopt_\link)}{\degasym_\link\stack_\link + \induced_\link}.
\end{equation}
Since we once again have that $\autoopt_\link + \humanopt_\link \neq 0$ we may simplify the above to
\begin{equation}
    \beta_\link \leq 1 - \frac{1-\degauto^*_\link(1-\degasym_\link)}{\gamma_\link^{-1}-(1-\degasym_\link)\stack_\link/(\autoopt_\link+\humanopt_\link)}.
\end{equation}
If we again take $\stack_\link = \degauto(\autoopt_\link + \humanopt_\link)$, then
\begin{equation}
    \beta_\link \leq 1 - \frac{1 - \degauto^*_\link(1-\degasym_\link)}{\gamma_\link^{-1} - \degauto(1-\degasym_\link)}.
\end{equation}

\textbf{Case 2.} If $\tilde{\gamma}_\link \leq \gamma_\link \leq 1/\degauto$, or equivalently $(\degasym_\link\stack_\link+\induced_\link) - (\degasym_\link\autoopt_\link+\humanopt_\link) \leq 0$, it immediately follows that $\beta_\link(\gamma_\link) \leq 0$ since $h_\link(\degasym_\link\stack_\link+\induced_\link) + b_\link > 0$.
\end{proof}

While Proposition \ref{prop:beta_link} provides an upper bound on the ratio of latency functions, this value depends on an unknown network-based parameter $\degauto^*_{\link}$ and may not be determined a priori. We may relax the derived upper bound in such a way to be independent of the optimal flow parameters $\autoopt_\link, \humanopt_\link$ as follows.

\begin{proposition} \label{prop:beta_link_relaxed}
For a mixed-autonomy routing game $(\network, \demand, \latency)$ with affine latency functions, network autonomy fraction $\degauto$, and mixed-autonomy SCALE strategy $\stack = \degauto(\autoopt + \humanopt)$, the upper bound in (\ref{eq:beta_link}) may be relaxed to
\begin{equation} \label{eq:beta_link_relaxed}
    \beta_\link(\gamma_\link) \leq 
    \begin{cases}
        \hfil 1 &
        \mathrm{for} \quad \gamma_\link = 0 \\
        \hfil 1 - \frac{\degasym_\link}{\gamma_\link^{-1} - \degauto(1-\degasym_\link)} &
        \mathrm{for} \quad 0 < \gamma_\link < \gamma^+_\link \\
        \hfil 0 &
        \mathrm{for} \quad \gamma^+_\link \leq \gamma_\link \leq 1/\degauto
    \end{cases}
\end{equation}
where $\gamma^+_\link = [\degauto(1-\degasym_\link) + \degasym_\link]^{-1}$.
\end{proposition}
\begin{proof}
Since $\tilde{\gamma}_\link$ and $1 - (1 - \degauto^*_\link(1-\degasym_\link))/(\gamma_\link^{-1} - \degauto(1-\degasym_\link))$ are both increasing in $\degauto^*_\link$ and $0 \leq \degauto^*_\link \leq 1$ we may conclude that $\tilde{\gamma}_\link \leq [\degauto(1-\degasym_\link)+\degasym_\link]^{-1}$ and
\begin{equation}
    1 - \frac{1 - \degauto^*_\link(1-\degasym_\link)}{\gamma_\link^{-1} - \degauto(1-\degasym_\link)} \leq 1 - \frac{\degasym_\link}{\gamma_\link^{-1} - \degauto(1-\degasym_\link)}.
\end{equation}

\end{proof}

Now, we utilize this relaxed bound on $\beta_\link$ in order to determine an upper bound on the link parameter $\omega_\link$.

\begin{proposition} For a mixed-autonomy routing game $(\network, \demand, \latency)$ with affine latency functions, network autonomy fraction $\degauto$, and mixed-autonomy SCALE strategy $\stack = \degauto(\autoopt + \humanopt)$, we have
\begin{equation} \label{eq:omega_link_lambda}
    \omega_\link(\lambda) \leq 
    \begin{cases}
        \hfil \frac{1}{\degauto}(1-\lambda) &
        \mathrm{for} \quad 0 \leq \lambda < \lambda^+_\link \\
        \hfil \frac{1-\Delta_\link(\lambda)}{\degauto(1-\degasym_\link)} - \frac{\degasym_\link(1-\Delta_\link(\lambda))^2}{\degauto^2(1-\degasym_\link)^2\Delta_\link(\lambda)}\lambda &
        \mathrm{for} \quad \lambda^+_\link \leq \lambda \leq 1
    \end{cases}
\end{equation}
where $\Delta_\link(\lambda) = \sqrt{\lambda\degasym_\link/(\degauto(1-\degasym_\link)+\lambda\degasym_\link)}$, and
\begin{equation}
    \lambda^+_{\link} = \frac{2\degasym_\link(1+\sqrt{1-\degauto}) - \degauto(1-\degasym_\link)\degasym_\link}{\degauto(1-\degasym_\link)^2 + 4\degasym_\link}.
\end{equation}
\end{proposition}
\begin{proof}
We prove this by calculating the value of $\gamma_\link$, dependent on $\lambda$, which maximizes the upper bound on $\omega_\link(\lambda)$.

Given (\ref{eq:omega_link}) and (\ref{eq:beta_link_relaxed}) we have that 
\begin{equation}
    \omega_\link(\lambda, \gamma_\link) \leq 
    \begin{cases}
        \hfil \gamma_\link &
        \gamma_\link = 0 \\
        \gamma_\link(1 - \frac{\degasym_\link}{\gamma_\link^{-1} - \degauto(1-\degasym_\link)}\lambda) &
        \mathrm{for} \quad 0 < \gamma_\link < \gamma^+_\link \\
        \hfil \gamma_\link(1-\lambda) &
        \mathrm{for} \quad \gamma^+_\link \leq \gamma_\link \leq 1/\degauto
    \end{cases}.
\end{equation}
In order to determine an upper bound for $\omega_\link(\lambda) = \sup_{\gamma_\link \in [0,1/\degauto]} \omega_\link(\lambda, \gamma_\link)$ we define 
\begin{equation}
    \omega_{\link,1}(\lambda,\gamma_\link) := \gamma_\link\left(1 - \frac{\degasym_\link}{\gamma_\link^{-1} - \degauto(1-\degasym_\link)}\lambda\right)
\end{equation}
 and 
\begin{equation}
     \omega_{\link,2}(\lambda,\gamma_\link) := \gamma_\link(1-\lambda).
\end{equation}
Additionally, take
\begin{equation}
    \omega_{\link,1}(\lambda) := \sup_{\gamma_\link \in (0,\gamma^+_\link)} \omega_{\link,1}(\lambda, \gamma_\link)
\end{equation}
and
\begin{equation}
    \omega_{\link,2}(\lambda) := \sup_{\gamma_\link \in [\gamma^+_\link, 1/\degauto]} \omega_{\link,2}(\lambda, \gamma_\link).
\end{equation}
Therefore, 
\begin{equation} \label{eq:omega_max}
    \omega_\link(\lambda) \leq \max\{0, \omega_{\link,1}(\lambda), \omega_{\link,2}(\lambda)\}.
\end{equation}

Graphs of $\omega_{\link,1}(\lambda, \gamma_\link)$, $\omega_{\link,2}(\lambda, \gamma_\link)$, and their relative maxima are shown in Fig.~\ref{fig:omega vs gamma}, which provide intuition for the following arguments.

\begin{figure}
    \centering
    \includegraphics[width=.29\textwidth]{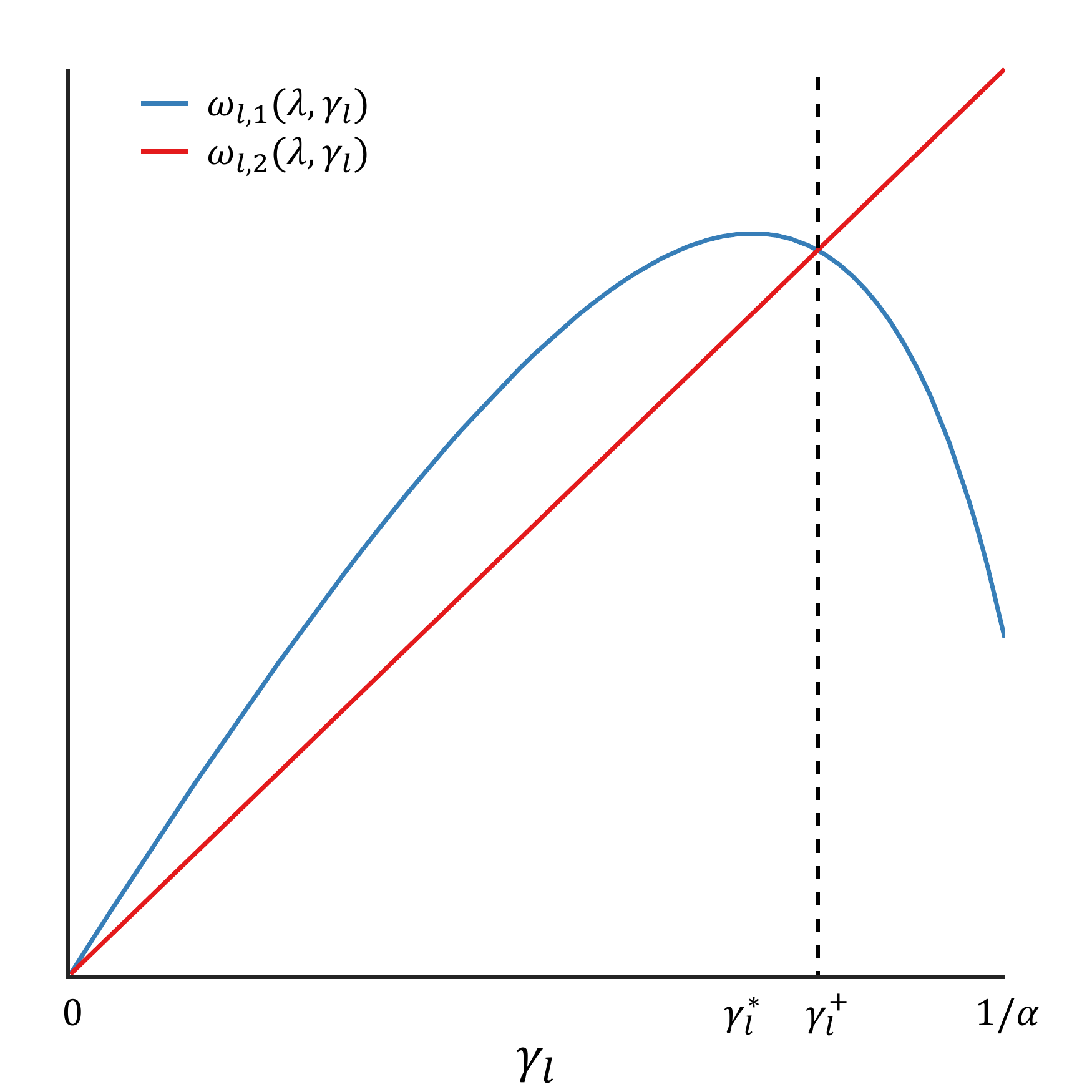}\hfil
    \caption{Characteristic plot of $\omega_{\link,1}(\lambda, \gamma_\link)$ and $\omega_{\link,2}(\lambda, \gamma_\link)$ for fixed $\degauto$, $\degasym_\link$, and $\lambda$. Clearly the supremum of $\omega_{\link,1}(\lambda, \gamma_\link)$ on $\gamma_\link(0,\gamma^+_\link)$ occurs at $\gamma_\link = \gamma^*_\link$ or in the limit $\gamma \to \gamma^+_\link$, dependent on network parameters. Additionally, the supremum of $\omega_{\link,2}(\lambda, \gamma_\link)$ occurs at $\gamma_\link = 1/\degauto$.}
    \label{fig:omega vs gamma}
\end{figure}

We will first turn our attention to calculating $\omega_{\link,1}(\lambda)$. Since $\omega_{\link,1}(\lambda, \gamma_\link)$ is continuous we have that the supremum either occurs in the limiting boundary values or critical values of $\gamma_\link$. We can immediately determine $\lim_{\gamma_\link \to 0} \omega_{\link,1}(\lambda, \gamma_\link) = 0$ and $\lim_{\gamma_\link \to \gamma^+_\link} \omega_{\link,1}(\lambda, \gamma_\link) = (1-\lambda)/(\degauto(1-\degasym_\link)+\degasym_\link)$.
Note that $\lim_{\gamma_\link \to \gamma^+_\link} \omega_{\link,1}(\lambda, \gamma_\link) \geq 0$ for all $\lambda \in [0,1]$.

If $\degasym_\link \neq 1$, the critical points exist and are given as values of $\gamma_\link$ such that $\partial/\partial \gamma_\link [\omega_{\link,1}(\lambda, \gamma_\link)] = 0$. It can be shown that the only such possible value of $\gamma_\link$ in our region of interest is given by $\gamma^*_\link = (1 - \Delta_\link(\lambda))/(\degauto(1-\degasym_\link))$
where $\Delta_\link(\lambda) = \sqrt{\lambda\degasym_\link/(\degauto(1-\degasym_\link)+\lambda\degasym_\link)}$. 
Additionally, it may be shown that $\partial^2/\partial \gamma_\link^2[\omega_{\link,1}(\lambda, \gamma_\link)]$ is negative on $\gamma_\link \in (0, \gamma^+_\link)$, and thus $\gamma^*_\link$ corresponds to a local maximum. This maximum is given by
\begin{equation}
    \omega_{\link,1}(\lambda, \gamma^*_\link) = \frac{1-\Delta_\link(\lambda)}{\degauto(1-\degasym_\link)} - \frac{\degasym_\link(1-\Delta(\lambda))^2}{\degauto^2(1-\degasym_\link)^2\Delta_\link(\lambda)}\lambda.
\end{equation}
Now, we just need to determine whether or not $\gamma^*_\link \in (0, \gamma^+_\link)$. If $\gamma^+_\link = \gamma^*_\link$, then
\begin{equation}
    \frac{1}{\degauto(1-\degasym_\link)+\degasym_\link} = \frac{1 - \Delta_\link(\lambda)}{\degauto(1-\degasym_\link)}
\end{equation}
and, once again, since $\degauto \neq 0$ and $\degasym_\link \neq 1$ the above statement is satisfied by
\begin{equation}
    \lambda_{\link,1} := \frac{\degasym_\link}{\degauto(1-\degasym_\link)+2\degasym_\link}
\end{equation}
Therefore, it follows that
\begin{equation}
    \omega_{\link,1}(\lambda) \leq 
    \begin{cases}
        \hfil \frac{1-\lambda}{\degauto(1-\degasym_\link)+\degasym_\link} &
        \mathrm{for} \quad 0 \leq \lambda \leq \lambda_{\link,1} \\
        \hfil \frac{1}{\degauto(1-\degasym_\link)} + \frac{2\degasym_\link(1-\Delta^{-1}_\link(\lambda))}{\degauto^2(1-\degasym_\link)^2}\lambda &
        \mathrm{for} \quad \lambda_{\link,1} < \lambda \leq 1.
    \end{cases}
\end{equation}

Since $\omega_{\link,2}(\lambda, \gamma_\link)$ is an increasing linear function of $\gamma_\link$, it immediately follows that
\begin{equation}
    \omega_{\link,2}(\lambda) = \frac{1}{\degauto}(1-\lambda).
\end{equation}


Now that we have expressions for $\omega_{\link,1}(\lambda)$ and $\omega_{\link,2}(\lambda)$, we may proceed to find $\omega_\link(\lambda)$ by determining the piecewise maximum as in~\eqref{eq:omega_max}. Fig.~\ref{fig:omega vs lambda} shows the relationship between $\omega_{\link,1}(\lambda)$ and $\omega_{\link,2}(\lambda)$.

We see that $\omega_{\link,1}(\lambda)$ on $\lambda \in [0,\lambda_{\link,1}]$ shares the same linear form with $\omega_{\link,2}(\lambda)$, and since $[\degauto(1-\degasym_\link) + \degasym_\link]^{-1} < 1/\degauto$, we have that $\omega_{\link,1}(\lambda) < \omega_{\link,2}(\lambda) \quad \mathrm{for} \quad 0 \leq \lambda < \lambda_{\link,1}$.

To determine the value of  $\omega_\link(\lambda)$ for $\lambda \in [\lambda_{\link,1},1]$ we can examine the intersection points of $\omega_{\link,1}(\lambda)$ and $\omega_{\link,2}(\lambda)$ in $[\lambda_{\link,1},1]$. More specifically, we are interested in values of $\lambda$ such that $\omega_{\link,1}(\lambda, \gamma^*_\link) = \omega_{\link,2}(\lambda)$. The only two values for which this occur are given as
\begin{equation}
    \lambda^-_\link = \frac{2\degasym_\link(1-\sqrt{1-\degauto}) - \degauto(1-\degasym_\link)\degasym_\link}{\degauto(1-\degasym_\link)^2 + 4\degasym_\link}
\end{equation}
and
\begin{equation}
    \lambda^+_\link = \frac{2\degasym_\link(1+\sqrt{1-\degauto}) - \degauto(1-\degasym_\link)\degasym_\link}{\degauto(1-\degasym_\link)^2 + 4\degasym_\link}.
\end{equation}
Moreover, it may be shown that for $\degauto \in (0,1)$ and $\degasym_\link \in (0,1]$, $\lambda^-_\link \in (0,\lambda_{\link,1})$ and $\lambda^+_\link \in (\lambda_{\link,1},1)$. Now, since $\omega_{\link,1}(1) = (1-\Delta_\link(1))/(\degauto(1-\degasym_\link)) > 0$ and $\omega_{\link,2}(1) = 0$, we have that $\omega_{\link,1}(1) > \omega_{\link,2}(1)$. Finally, since $\omega_{\link,1}(\lambda_{\link,1}) < \omega_{\link,2}(\lambda_{\link,1})$ and $\omega_{\link,1}(1) > \omega_{\link,2}(1)$ with only one intersection point in $(\lambda_{\link,1},1)$ given by $\lambda = \lambda^+_\link$, it must follow that
\begin{equation}
    \omega_{\link,1}(\lambda) < \omega_{\link,2}(\lambda) \quad \mathrm{for} \quad \lambda_{\link,1} \leq \lambda < \lambda^+_\link
\end{equation}
and
\begin{equation}
    \omega_{\link,1}(\lambda) \geq \omega_{\link,2}(\lambda) \quad \mathrm{for} \quad \lambda^+_\link \leq \lambda \leq 1.
\end{equation}

\begin{figure}
    \centering
    \includegraphics[width=.29\textwidth]{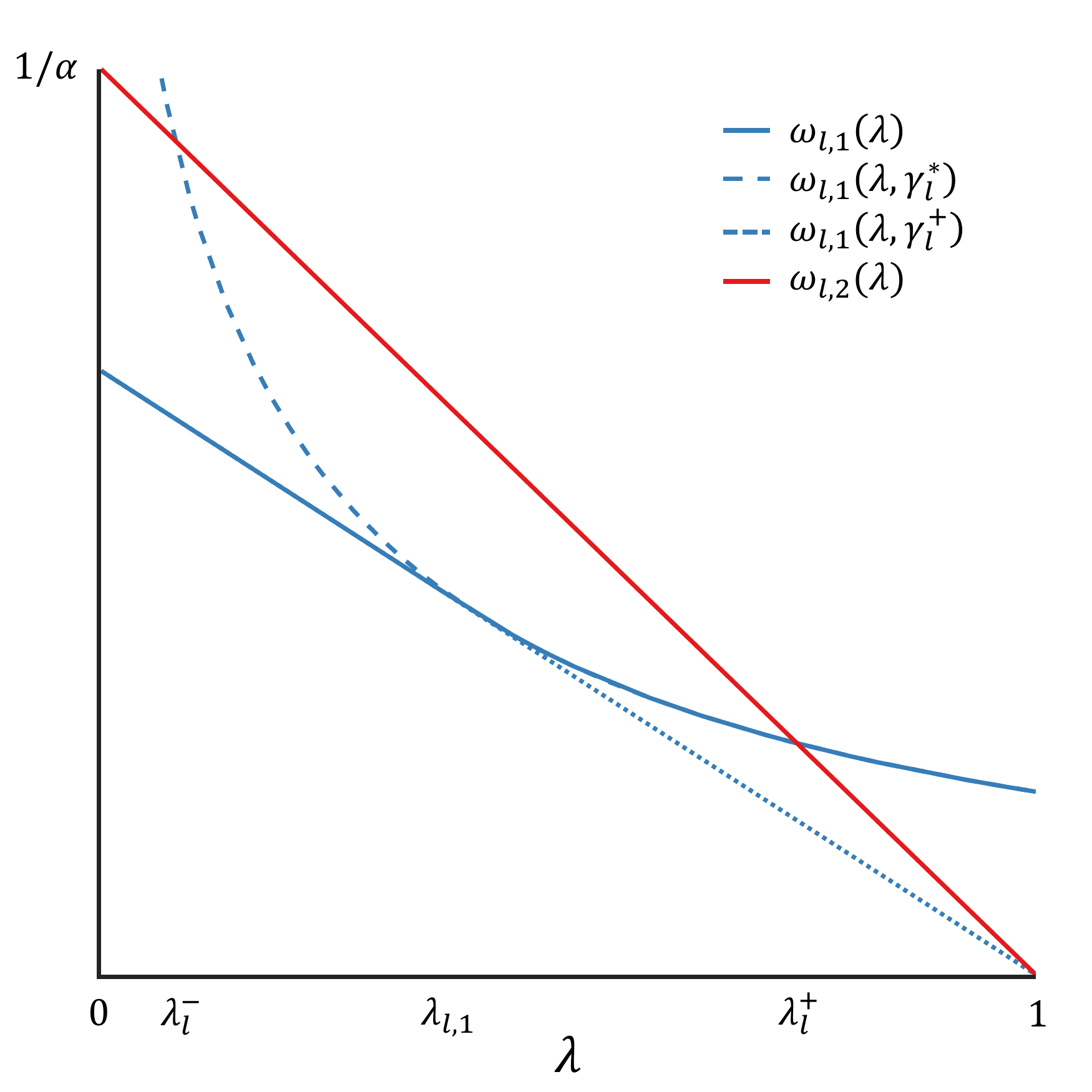}\hfil
    \caption{Characteristic plot of $\omega_{\link,1}(\lambda)$ and $\omega_{\link,2}(\lambda)$ for fixed $\degauto$ and $\degasym_\link$. Here it can be seen that the intersection point $\lambda^-_\link$ lies in $(0,\lambda_{\link,1})$ and the intersection point $\lambda^+_\link$ lies in $(\lambda_{\link,1},1)$. This leads to the conclusion that $\omega_\link(\lambda) = \omega_{\link,2}(\lambda)$ for $\lambda \in [0,\lambda^+_\link)$ and $\omega_\link(\lambda) = \omega_{\link,1}(\lambda)$ for $\lambda \in [\lambda^+_\link,1]$.}
    \label{fig:omega vs lambda}
\end{figure}

\end{proof}

Up to this point we have determined the value of the link flow ratio $\gamma_\link = (\autoopt_\link + \humanopt_\link)/(\stack_\link + \induced_\link)$ that maximizes the upper bound on $\omega_\link(\lambda)$. Ultimately, in order to calculate bounds on the price of anarchy, we want to find $\omega(\lambda) = \max_{\link \in \linkset}\omega_\link(\lambda)$. The only link dependent parameter remaining in~\eqref{eq:omega_link_lambda} is $\degasym_\link$. That is, we would like to find the value $\degasym_\link$ for $\link \in \linkset$ that maximizes $\omega_\link(\lambda)$. 
It can be shown that if we take $\degasym_\link = \degasym$ for all $\link \in \linkset$ where
\begin{equation}
    \degasym = \min_{\link \in \linkset} \degasym_\link
\end{equation}
we recover
\begin{equation} \label{eq:omega_lambda}
    \omega(\lambda) \leq 
    \begin{cases}
        \hfil \frac{1}{\degauto}(1-\lambda) &
        \mathrm{for} \quad 0 \leq \lambda \leq \lambda^+ \\
        \hfil \frac{1-\Delta(\lambda)}{\degauto(1-\degasym)} - \frac{\degasym(1-\Delta(\lambda))^2}{\degauto^2(1-\degasym)^2\Delta(\lambda)}\lambda &
        \mathrm{for} \quad \lambda^+ < \lambda \leq 1
    \end{cases}
\end{equation}
where $\Delta(\lambda) = \sqrt{\lambda\degasym/(\degauto(1-\degasym)+\lambda\degasym)}$ and \begin{equation}
    \lambda^+ = \frac{2\degasym(1+\sqrt{1-\degauto}) - \degauto(1-\degasym)\degasym}{\degauto(1-\degasym)^2 + 4\degasym}.
\end{equation}

In order to implement the $\lambda$-approach detailed in Proposition \ref{prop:lambda_method}, we need to determine the set $\Lambda$ given by (\ref{eq:Lambda_set}). We can express this set in terms of $\lambda_{\omega_1} := \{\lambda \mid \omega_1(\lambda) = 1\}$ and $\lambda_{\omega_2} := \{\lambda \mid \omega_2(\lambda) = 1\}$:
\begin{equation}
    \lambda_{\omega_1} = \frac{(1-\degauto(1-\degasym))^2}{4 \degasym}
\end{equation}
and
\begin{equation}
    \lambda_{\omega_2} = 1-\degauto.
\end{equation}
Since we know the functional form of the $\omega$ parameter given by (\ref{eq:omega_lambda}), we can determine that there are three separate cases for the feasible set $\Lambda$:

\textbf{Case 1.} If $\omega(1) \geq 1$, then $\Lambda = \emptyset$ and no value of $\lambda \in [0,1]$ yields an admissible bound on the price of anarchy. This condition is equivalent to $\omega_1(1) \geq 1$, which is true for $\degauto \leq \degauto_0$ where $\degauto_0 := (1-2\sqrt{\degasym})/(1-\degasym)$.

\textbf{Case 2.} If $\omega(\lambda^+) > 1$ and $\omega(1) < 1$, then $\omega(\lambda) < 1$ for $\lambda > \lambda_{\omega_1}$ and $\Lambda = (\lambda_{\omega_1}, 1]$. This case holds for $\degauto < \tilde{\degauto}$ and $\degauto > \degauto_0$ where $\tilde{\degauto} := (1-2\sqrt{\degasym})/(1-\sqrt{\degasym})^2$.

\textbf{Case 3.} If $\omega(\lambda^+) \leq 1$, then $\omega(\lambda) < 1$ for $\lambda > \lambda_{\omega_2}$ and $\Lambda = (\lambda_{\omega_2}, 1]$. Clearly, this is true for $\degauto \geq \tilde{\degauto}$.

Thus, we may describe the feasible sets $\Lambda$ based on constraint sets of $\degauto$. That is,
\begin{equation} \label{eq:Lambda}
    \Lambda = 
    \begin{cases}
        \hfil \emptyset & 
        \mathrm{for} \quad \degauto \in A_0 \\
        \hfil (\lambda_{\omega_1}, 1] &
        \mathrm{for} \quad \degauto \in A_{\omega_1} \\
        \hfil (\lambda_{\omega_2}, 1] &
        \mathrm{for} \quad \degauto \in A_{\omega_2}
    \end{cases}
\end{equation}
where
\begin{equation}
    A_0 := \{\degauto \mid \degauto \leq \degauto_0, \: 0 \leq \degauto \leq 1\}
\end{equation}
\begin{equation}
    A_{\omega_1} := \{\degauto \mid \degauto > \degauto_0, \: \degauto \leq \tilde{\degauto}, \: 0 \leq \degauto \leq 1\}
\end{equation}
\begin{equation}
    A_{\omega_2} := \{\degauto \mid \degauto > \tilde{\degauto}, \: 0 \leq \degauto \leq 1\}.
\end{equation}
\begin{remark}
    Since $\tilde{\degauto} \geq \degauto_0$ whenever both $\tilde{\degauto},\degauto_0 \in [0,1]$, the condition $\degauto > \tilde{\degauto}$ implies $\degauto > \degauto_0$.
\end{remark}

\begin{theorem} \label{th:PoA_constraint_sets}
For a mixed-autonomy routing game $(\network, \demand, \latency)$ with affine latency functions, network autonomy fraction $\degauto$, minimum degree of asymmetry $\degasym$, and mixed-autonomy SCALE strategy $\stack = \degauto(\autoopt + \humanopt)$, we have that the price of anarchy is bounded as a function of only $\degauto$ and $\degasym$ as follows:
\begin{equation}
    \mathrm{PoA} \leq 
    \begin{cases}
        \hfil{\infty} &
        \mathrm{for} \quad \degauto \in A_0 \\
        \hfil \mathrm{PoA}_{\omega_1}(1) &
        \mathrm{for} \quad \degauto \in A_1 \\
        \hfil \mathrm{PoA}_{\omega_1}(\lambda^*) &
        \mathrm{for} \quad \degauto \in A_{\lambda^*} \\
        \hfil \mathrm{PoA}_{\omega_1}(\lambda^+) & 
        \mathrm{for} \quad \degauto \in A_{\lambda^+} \\
    \end{cases},
\end{equation}
where
\begin{equation} 
    \resizebox{0.43\textwidth}{!}{
    $\mathrm{PoA}_{\omega_1}(1) = \frac{\degauto^2(1-\degasym)^2 - \degauto(1-\degasym) -2\degasym - 2\sqrt{\degasym^2+\degauto(1-\degasym)\degasym}}{\left(1-\degauto(1-\degasym)\right)^2 - 4\degasym}$,
    }
\end{equation}
\begin{equation}
    \mathrm{PoA}_{\omega_1}(\lambda^*) = \frac{1-\degauto(1-\degasym)}{\degasym},
\end{equation}
\begin{equation}  
    \resizebox{0.43\textwidth}{!}{
    $\mathrm{PoA}_{\omega_1}(\lambda^+) = \frac{2\degauto\degasym(1+\sqrt{1-\degauto}) - \degauto^2(1-\degasym)\degasym}{\degauto^2(1-\degasym)^2 + \degauto(5\degasym-1) - 2\degasym(1-\sqrt{1-\degauto})}$,
    }
\end{equation}
and
\begin{equation}
    A_0 = \{\degauto \mid \degauto \leq \degauto_0, \: 0 \leq \degauto \leq 1\},
\end{equation}
\begin{equation}
    A_1 = \{\degauto \mid \degauto > \degauto_0, \: \degauto < \degauto_1, \: 0 \leq \degauto \leq 1\},
\end{equation}
\begin{equation}
    A_{\lambda^*} = \{\degauto \mid \degauto \geq \degauto_1, \: \degauto < \degauto_2, \: 0 \leq \degauto \leq 1\},
\end{equation}
\begin{equation}
    A_{\lambda^+} = \{\degauto \mid \degauto \geq \degauto_2, \: 0 \leq \degauto \leq 1\},
\end{equation}
with
\begin{equation}
    \degauto_0 = \frac{1-2\sqrt{\degasym}}{1-\degasym}, \; \degauto_1 = \frac{1-2\sqrt{\degasym}}{(1-\sqrt{\degasym})^2}, \; \degauto_2 = \frac{1-2\degasym}{(1-\degasym)^2}.
\end{equation}
\end{theorem}
\begin{proof}
    From Proposition~\ref{prop:lambda_method}, we can determine the minimum bound induced on the price of anarchy by solving the following optimization problem:
\begin{equation}
    \mathrm{PoA} \leq \inf_{\lambda \in \Lambda} \frac{\lambda}{1-\omega(\lambda)}
\end{equation}
Since the set $\Lambda$ is piecewise in $\degauto$ and $\omega$ is piecewise continuous in $\lambda$, we may proceed with the following cases:

\textbf{Case 1.} If $\degauto \in A_0$, we have that $\Lambda = \emptyset$ and no value of $\lambda \in [0,1]$ yields an admissible bound on the price of anarchy.

\textbf{Case 2.} If $\degauto \in A_{\omega_1}$, we have that $\mathrm{PoA} \leq \inf_{\lambda \in (\lambda_{\omega_1}, 1]} \mathrm{PoA}_{\omega_1}(\lambda)$
where we define $\mathrm{PoA}_{\omega_i}(\lambda) := \lambda/(1-\omega_i(\lambda))$ for $i=1,2$.
Since $\lim_{\lambda \to \lambda_{\omega_1}} \mathrm{PoA}_{\omega_1}(\lambda) = \infty$, we have either the minimum occurs at the boundary value $\lambda = 1$ or a critical value $\lambda = \lambda^*$, which is defined as
\begin{equation}
    \lambda^* := \left\{\lambda \mathrel{\Big|} \frac{\partial}{\partial \lambda} \mathrm{PoA}_{\omega_1}(\lambda) = 0\right\}.
\end{equation}
The price of anarchy bound at the boundary is given by
\begin{equation} 
    \resizebox{0.43\textwidth}{!}{
        $\mathrm{PoA_{\omega_1}}(1) = \frac{\degauto^2(1-\degasym)^2 - \degauto(1-\degasym) - 2\degasym -2\sqrt{\degasym^2 + \degauto(1-\degasym)\degasym}}{(1-\degauto(1-\degasym))^2 - 4\degasym}$.
    }
\end{equation}
The unique value of $\lambda^*$ is given as
\begin{equation}
    \lambda^* = \frac{(1-\degauto(1-\degasym))^2}{[2-\degauto(1-\degasym)]\degasym}
\end{equation}
and results in a price of anarchy bound of
\begin{equation}
    \mathrm{PoA_{\omega_1}}(\lambda^*) = \frac{1 - \degauto(1-\degasym)}{\degasym}.
\end{equation}

Since $\mathrm{PoA}_{\omega_1}$ is a convex function of $\lambda$ on $(\lambda_{\omega_1}, 1]$, we have that the minimum is given by $\mathrm{PoA_{\omega_1}}(\lambda^*)$ if $\lambda^* \in (\lambda_{\omega_1}, 1]$ and $\mathrm{PoA_{\omega_1}}(1)$ otherwise. It can be shown that $\lambda^* \in (\lambda_{\omega_1}, 1]$ if and only if $\degauto \geq \degauto_1$ where
\begin{equation}
    \degauto_1 =  1 + \frac{\degasym - \sqrt{\degasym^2+4\degasym}}{2(1-\degasym)}.
\end{equation}
Hence, we may conclude
\begin{equation} \label{eq:PoA_A1}
    \mathrm{PoA} =
    \begin{cases}
        \hfil \mathrm{PoA_{\omega_1}}(1) &
        \mathrm{for} \quad \degauto \in (A_{\omega_1} \cap \degauto < \degauto_1) \\
        \hfil \mathrm{PoA_{\omega_1}}(\lambda^*) &
        \mathrm{for} \quad \degauto \in (A_{\omega_1} \cap \degauto \geq \degauto_1)
    \end{cases}
\end{equation}

\textbf{Case 3.} If $\degauto \in A_{\omega_2}$, we have that $\mathrm{PoA} \leq \inf_{\lambda \in (\lambda_{\omega_2}, 1]} \mathrm{PoA}_{\omega}(\lambda)$
where
\begin{equation}
    \mathrm{PoA}_{\omega}(\lambda) = 
    \begin{cases}
        \hfil \mathrm{PoA}_{\omega_2}(\lambda) & \mathrm{for} \quad \lambda_{\omega_2} \leq \lambda < \lambda^+ \\ 
        \hfil \mathrm{PoA}_{\omega_1}(\lambda) &
        \mathrm{for} \quad \lambda^+ \leq \lambda \leq 1
    \end{cases}.
\end{equation}

Since $\lim_{\lambda \to \lambda_{\omega_2}} \mathrm{PoA}_{\omega_2}(\lambda) = \infty$, we have either the minimum occurs at the boundary values $\lambda = \lambda^+, 1$ or critical values of $\mathrm{PoA}_{\omega_1}, \mathrm{PoA}_{\omega_2}$. Since $\partial \mathrm{PoA}_{
\omega_2}/\partial \lambda$ is nonzero everywhere, we have that the only critical value is given by $\lambda^*$ as described in the previous case. Thus, the last bound of interest is
\begin{equation}
    \mathrm{PoA}_{\omega_1}(\lambda^+) = \frac{2\degauto\degasym(1+\sqrt{1-\degauto}) - \degauto^2(1-\degasym)\degasym}{\degauto^2(1-\degasym)^2 + \degauto(5\degasym-1) - 2\degasym(1+\sqrt{1-\degauto})}.
\end{equation}

Now, since $\mathrm{PoA}_{\omega_1}$ and $\mathrm{PoA}_{\omega_2}$ are convex functions of $\lambda$ we have that the minimum can be directly determined by the value of $\lambda^*$ relative to $\lambda^+$ and $1$. This provides three subcases to examine:

\begin{enumerate}[label=\alph*.)]
    \item If $\lambda^* > 1$, then $\inf_{\lambda \in (\lambda_{\omega_2}, 1]} \mathrm{PoA}_{\omega}(\lambda) = \mathrm{PoA}_{\omega_1}(1)$. Such a constraint on $\lambda^*$ is satisfied for $\degauto < \degauto_1$.
    \item If $\lambda^* \leq \lambda^+$, then $\inf_{\lambda \in (\lambda_{\omega_2}, 1]} \mathrm{PoA}_{\omega}(\lambda) = \mathrm{PoA}_{\omega_1}(\lambda^+)$. This occurs for $\degauto \geq \degauto_2$ where $\degauto_2 := \{\degauto \mid \lambda^+ = \lambda^* \}$ is given by $\degauto_2 = (1-2\degasym)/(1-\degasym)^2$
    \item If $\lambda^+ < \lambda^* \leq 1$, then $\inf_{\lambda \in (\lambda_{\omega_2}, 1]} \mathrm{PoA}_{\omega}(\lambda) = \mathrm{PoA}_{\omega_1}(\lambda^*)$. This case holds if $\degauto \geq \degauto_1$ and $\degauto < \degauto_2$.
\end{enumerate}

Thus, if $\degauto \in A_{\omega_2}$,
\begin{equation} \label{eq:PoA_A2}
    \resizebox{0.49\textwidth}{!}{ 
    $\mathrm{PoA} \leq
    \begin{cases}
        \hfil \mathrm{PoA_{\omega_1}}(1) &
        \mathrm{for} \quad \degauto \in (A_{\omega_2} \cap \degauto < \degauto_1) \\
        \hfil \mathrm{PoA_{\omega_1}}(\lambda^*) &
        \mathrm{for} \quad \degauto \in (A_{\omega_2} \cap \degauto \geq \degauto_1 \cap \degauto < \degauto_2) \\ 
        \hfil \mathrm{PoA_{\omega_1}}(\lambda^+) &
        \mathrm{for} \quad \degauto \in (A_{\omega_2} \cap \degauto \geq \degauto_2) \\ 
    \end{cases}$.
    }
\end{equation}

Now, since we have constraint regions of $\degauto$ that correspond to specific values of the price of anarchy bound $\mathrm{PoA_{\omega_1}}(1)$, $\mathrm{PoA_{\omega_1}}(\lambda^*)$, and $\mathrm{PoA_{\omega_1}}(\lambda^+)$ given by (\ref{eq:PoA_A1}) and (\ref{eq:PoA_A2}), we may take a union over the matching constraint sets and preserve the inequalities. Therefore, if we define
\begin{equation}
    A_1 := (A_{\omega_1} \cap \degauto < \degauto_1) \cup (A_{\omega_2} \cap \degauto < \degauto_1),
\end{equation}
\begin{equation}
    A_{\lambda^*} := (A_{\omega_1} \cap \degauto \geq \degauto_1) \cup (A_{\omega_2} \cap \degauto \geq \degauto_1 \cap \degauto < \degauto_2),
\end{equation}
\begin{equation}
    A_{\lambda^+} := (A_{\omega_2} \cap \degauto \geq \degauto_2),
\end{equation}
then we have
\begin{equation} \label{eq:PoA_A}
    \mathrm{PoA} \leq
    \begin{cases}
        \hfil \mathrm{PoA_{\omega_1}}(1) &
        \mathrm{for} \quad \degauto \in A_1 \\
        \hfil \mathrm{PoA_{\omega_1}}(\lambda^*) &
        \mathrm{for} \quad \degauto \in A_{\lambda^*} \\ 
        \hfil \mathrm{PoA_{\omega_1}}(\lambda^+) &
        \mathrm{for} \quad \degauto \in A_{\lambda^+} \\ 
    \end{cases}.
\end{equation}
Finally, it can be shown that theses constraint sets reduce to
\begin{equation}
    A_1 = \{\degauto \mid \degauto > \degauto_0, \: \degauto < \degauto_1, \: 0 \leq \degauto \leq 1\},
\end{equation}
\begin{equation}
    A_{\lambda^*} = \{\degauto \mid \degauto \geq \degauto_1, \: \degauto < \degauto_2, \: 0 \leq \degauto \leq 1\},
\end{equation}
\begin{equation}
    A_{\lambda^+} = \{\degauto \mid \degauto \geq \degauto_2, \: 0 \leq \degauto \leq 1\}.
\end{equation}

\end{proof}

The constraint sets from Theorem~\ref{th:PoA_constraint_sets} are shown in Fig.~\ref{fig:constraint_sets}. By calculating the constraint sets given in Theorem~\ref{th:PoA_constraint_sets} for the specific regions of $\degasym$, we may directly obtain the following result.

\begin{figure}
    \centering
    \includegraphics[width=.3\textwidth]{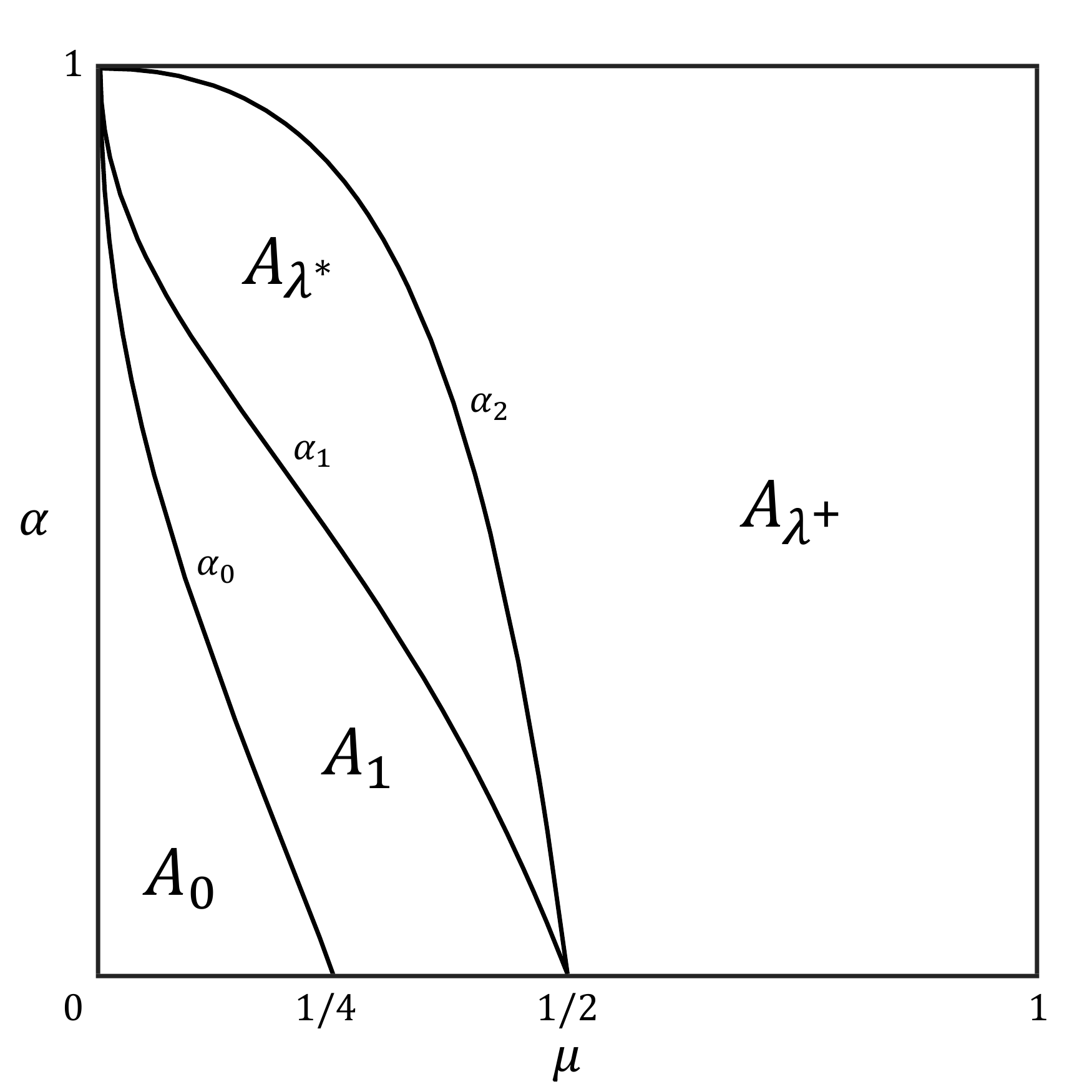}
    \caption{\small The constraint sets as described in Theorem~\ref{th:PoA_constraint_sets}. We see that the sets $A_{\lambda^*}, A_1$ are nonempty only for $0 < \degasym < 1/2$ and $A_0$ is nonempty only for $0 < \degasym \leq 1/4$.}
    \label{fig:constraint_sets}
    \vspace{-0.2in}
\end{figure}

\begin{corollary} \label{co:PoA_bounds}
For a mixed-autonomy routing game $(\network, \demand, \latency)$ with affine latency functions, network autonomy fraction $\degauto$, minimum degree of asymmetry $\degasym$, and mixed-autonomy SCALE strategy $\stack = \degauto(\autoopt + \humanopt)$, we have that the price of anarchy is bounded as a function of only $\degauto$ and $\degasym$ as follows:
\begin{enumerate}[label=\alph*)]
    \item If $1/2 \leq \degasym \leq 1$, then
    \begin{equation}
        \mathrm{PoA} \leq \mathrm{PoA}_{\omega_1}(\lambda^+) \quad \mathrm{for} \quad 0 \leq \degauto \leq 1.
    \end{equation}
    \item If $1/4 < \degasym < 1/2$, then
    \begin{equation}
        \resizebox{0.39\textwidth}{!}{ 
        $\mathrm{PoA} \leq
            \begin{cases}
                \hfil \mathrm{PoA}_{\omega_1}(1) &
                \mathrm{for} \quad 0 \leq \degauto < \degauto_1 \\
                \hfil \mathrm{PoA}_{\omega_1}(\lambda^*) &
                \mathrm{for} \quad \degauto_1 \leq \degauto < \degauto_2 \\
                \hfil \mathrm{PoA}_{\omega_1}(\lambda^+) &
                \mathrm{for} \quad \degauto_2 \leq \degauto \leq 1
            \end{cases}$.
        }
    \end{equation}
    \item If $0 < \degasym < 1/4$, then
    \begin{equation}
        \resizebox{0.39\textwidth}{!}{
        $\mathrm{PoA} \leq
            \begin{cases}
                \hfil \infty &
                \mathrm{for} \quad 0 \leq \degauto \leq \degauto_0 \\
                \hfil \mathrm{PoA}_{\omega_1}(1) &
                \mathrm{for} \quad \degauto_0 < \degauto < \degauto_1 \\
                \hfil \mathrm{PoA}_{\omega_1}(\lambda^*) &
                \mathrm{for} \quad \degauto_1 \leq \degauto < \degauto_2 \\
                \hfil \mathrm{PoA}_{\omega_1}(\lambda^+) &
                \mathrm{for} \quad \degauto_2 \leq \degauto \leq 1
            \end{cases}$.
            }
    \end{equation}
\end{enumerate}
\end{corollary}

The price of anarchy bounds from Corollary~\ref{co:PoA_bounds} are shown in Fig.~\ref{fig:PoA Bounds}.

\begin{figure*}
    \centering
    \includegraphics[width=0.87\textwidth]{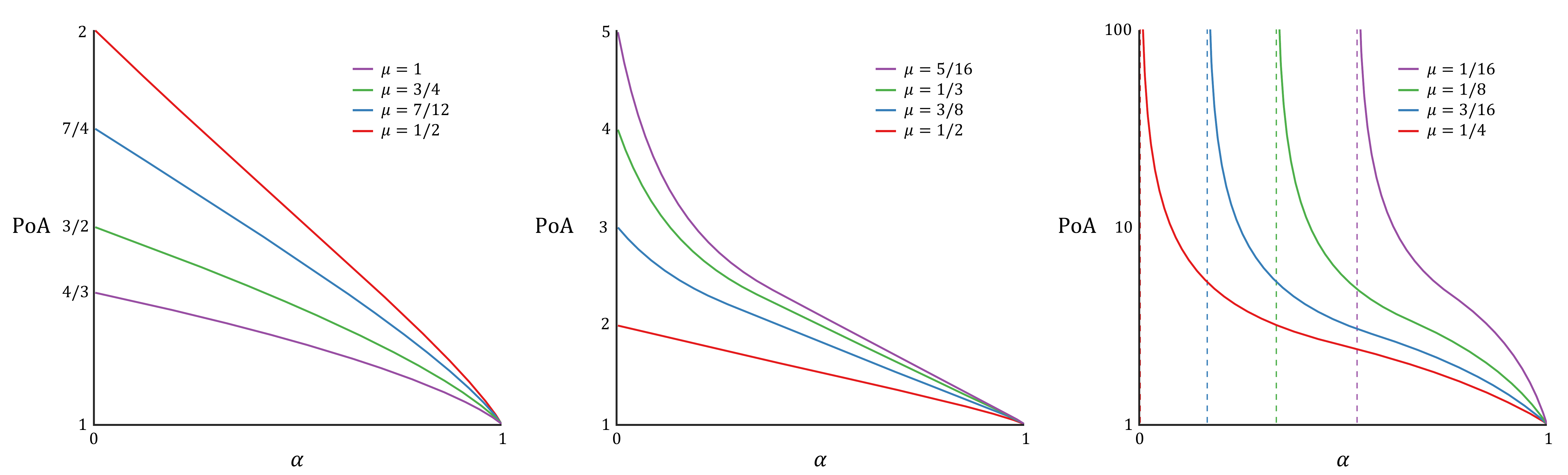}
    \caption{(Left) Price of anarchy upper bounds for select values of $\degasym \in [1/2,1]$. The bound for $\degasym=1$ is equivalent to the bound shown in \cite{BONIFACI2010}. (Middle) Upper bounds for $\degasym \in (1/4, 1/2)$. $\degasym=1/2$ is shown as well for comparison with the leftmost plot. (Right) Upper bounds for $\degasym \in (0,1/4]$, with corresponding values of ${\degauto}_0$ graphed as dashed vertical lines. Notice the logarithmic scale of the y-axis used to visualize the asymptotic behavior as $\degauto \to \degauto_0$}
    \label{fig:PoA Bounds}
    \vspace{-3mm}
\end{figure*}

Now we show that our price of anarchy upper bounds reduce to expected values and results from previous works:

\begin{enumerate}[label=\alph*.)]
    \item For a fully centrally routed autonomous network, i.e. $\degauto = 1$, we expect the price of anarchy to reduce to 1 since the central planner can route all of the vehicles to the optimal solution.
    Consider $\lim_{\degauto \to 1} \mathrm{PoA}$. For all cases in Theorem~\ref{th:PoA_constraint_sets}, $\degauto = 1$ belongs to $A_{\lambda^+}$. Thus, this limit is given as $\lim_{\degauto \to 1} \mathrm{PoA_{\omega_1}}(\lambda^+) = 1.$
    \item The implementation of a central planner should decrease the price of anarchy within the network. For our bound, consider $\lim_{\degauto \to 0} \mathrm{PoA}$. By Theorem~\ref{th:PoA_constraint_sets} we know that this limit is given by $\lim_{\degauto \to 0}\mathrm{PoA_{\omega_1}}(1)$ for $1/4 < \degasym < 1/2$ and $\lim_{\degauto \to 0}\mathrm{PoA_{\omega_1}}(\lambda^+)$ for $1/2 \leq \degasym \leq 1$. It can be shown that
    \begin{equation}
         \lim_{\degauto \to 0}\mathrm{PoA_{\omega_1}}(1) = \lim_{\degauto \to 0}\mathrm{PoA_{\omega_1}}(\lambda^+) = \frac{4\degasym}{4\degasym - 1}.
    \end{equation}
    Since we have that our prince of anarchy bound is a decreasing function of $\degauto$, we have that $\mathrm{PoA} < 4\degasym / (4\degasym - 1)$ for all $\degauto \in (0,1)$ and $\degasym \in (1/4, 1]$. Hence, the mixed autonomy SCALE strategy \textit{always reduces} the price of anarchy bound relative to the autonomous selfish routing bound developed in \cite{lazar2018price}.
    \item For a network with a single class of vehicles, that is, $\degasym = 1$, we should have that our bound reduces to previously known results. To this end, we calculate the upper bound on $\mathrm{PoA}\Bigr|_{\degasym = 1}$. From Theorem~\ref{th:PoA_constraint_sets} we have that this value is
    \begin{equation}
        \mathrm{PoA_{\omega_1}}(\lambda^+) \Bigr|_{\degasym = 1} = \frac{(1+\sqrt{1-\degauto})^2}{2(1+\sqrt{1-\degauto})-1}.
    \end{equation}
    This aligns with the price of anarchy bound given in \cite{BONIFACI2010}.
\end{enumerate}

\section{Conclusion and Future Work}\label{sec:conclusion}
In this paper, we developed a Stackelberg routing strategy for autonomous cars in a mixed-autonomy traffic network with arbitrary geometry. We developed an extension of the SCALE Stackelberg strategy for mixed-autonomy networks with affine latency functions. We bounded the price of anarchy that our Stackelberg strategy induces and proved that our proposed Stackelberg routing will reduce the network price of anarchy. We demonstrated that in the limit, when the fraction of autonomous cars is negligible, our bound recovers the price of anarchy bounds for networks of only human-driven cars. In the future, we would like to examine the tightness of our bounds. While our analysis here was limited to affine latency functions, the methodology proposed in Section~\ref{sec:prior-work} is applicable to general nonlinear latency functions. We would like to further extend our analysis to include results for mixed-autonomy routing games with such latency functions.

\section{Acknowledgment}
This work is supported by the National Science Foundation, under CAREER Award
ECCS-2145134.


 \bibliographystyle{ieeetr}%
\begingroup
\setstretch{0.87}
\bibliography{bib/references}
\endgroup

\end{document}